\documentclass[11pt,letterpaper]{article}
\usepackage[margin=1in]{geometry}
\usepackage{amsmath,amssymb,amsthm,mathtools}
\usepackage{booktabs}
\usepackage{float}
\usepackage{enumitem}
\usepackage{xcolor}
\usepackage{tikz}
\usetikzlibrary{patterns}
\usepackage[colorlinks=true,linkcolor=blue!50!black,citecolor=green!40!black,urlcolor=blue!50!black]{hyperref}
\usepackage[capitalise,nameinlink,noabbrev]{cleveref}

\newtheorem{theorem}{Theorem}[section]
\newtheorem{lemma}[theorem]{Lemma}
\newtheorem{proposition}[theorem]{Proposition}
\newtheorem{corollary}[theorem]{Corollary}
\newtheorem{fact}[theorem]{Fact}

\newcommand{\R}{\mathbb{R}}
\newcommand{\E}{\mathbb{E}}
\newcommand{\F}{\mathcal{F}}
\newcommand{\1}{\mathbf{1}}
\newcommand{\Fnorm}[1]{\lVert #1\rVert_\infty}
\newcommand{\eps}{\varepsilon}
\newcommand{\osc}{\operatorname{osc}}
\newcommand{\midp}{\operatorname{mid}}
\newcommand{\Rin}{R^{\mathrm{in}}}
\newcommand{\pmone}{\{\pm1\}}

\makeatletter\@secpenalty=0\makeatother
\newcommand{\AIsections}{\cref{sec:proof,app:facts}}

\title{Tight Bounds for Tusn\'ady's Problem in the Plane}
\author{Zhewei Wei\\[2pt]
\normalsize Gaoling School of Artificial Intelligence, Renmin University of China}
\date{}

\usepackage{etoolbox}
\apptocmd{\thebibliography}{\setlength{\itemsep}{2pt}\setlength{\parsep}{0pt}}{}{}

\usetikzlibrary{decorations.pathreplacing}

\begin{document}
\hypersetup{pageanchor=false}
\maketitle

\begin{abstract}
We show that the worst-case combinatorial discrepancy of $n$ points in the plane with respect to axis-parallel rectangles is $\Theta(\log^{3/2}n)$. The known bounds were $\Omega(\log n)$ and $O(\log^{3/2}n)$; we prove the matching lower bound. It holds for random point sets: for every $A>0$, there is a constant $c_A>0$ such that, with probability at least $1-e^{-An}$, every coloring of $n$ independent uniform points in the unit square has an anchored rectangle with imbalance at least $c_A(\log_2n)^{3/2}$. The proof is surprisingly simple and elementary. It reveals one coordinate digit by digit. With overwhelming probability over the points, the conditional gains of an oscillation potential add up to $\Omega(\log^{3/2}n)$ over $\Theta(\log n)$ digits. Bounded differences control the fluctuations well enough for a union bound over all colorings.
\end{abstract}
\thispagestyle{empty}
\newpage
\setcounter{page}{1}
\hypersetup{pageanchor=true}

\section{Introduction}\label{sec:intro}

Given a set $P$ of $n$ points in the plane, we would like to color each point by $+1$ or $-1$ so that every axis-parallel rectangle contains nearly the same number of points of each color. For a coloring $\chi:P\to\pmone$, the \emph{imbalance} of a rectangle $R$ is $\bigl|\sum_{p\in P\cap R}\chi(p)\bigr|$. The \emph{discrepancy} of $P$ is the minimum, over all colorings, of the maximum imbalance over all axis-parallel rectangles. Write $\Delta_d(n)$ for the largest discrepancy of an $n$-point set in $\R^d$ with respect to axis-parallel boxes. \emph{Tusn\'ady's problem in the plane} asks for the asymptotic order of $\Delta_2(n)$.

\paragraph{Main result.}
Consider $n$ points $P_i=(U_i,V_i)$ in $[0,1]^2$ and a coloring $\chi=(\chi_1,\dots,\chi_n)\in\pmone^n$, where $\chi_i$ is the color of $P_i$. Let $F_\chi(x,y)=\sum_{i=1}^n\chi_i\,\1\{P_i\in[0,x]\times[0,y]\}$, so that $|F_\chi(x,y)|$ is the imbalance of the anchored rectangle $[0,x]\times[0,y]$. We write $\Fnorm{F_\chi}$ for the supremum of $|F_\chi|$ over $[0,1]^2$. Logarithms are natural unless a base is specified.

\begin{theorem}\label{thm:main}
For every $A>0$, there is a constant $c_A>0$ such that the following holds for every integer $n\ge1$. If $P_1,\dots,P_n$ are independent uniform points in $[0,1]^2$, then
\[
\Pr\Bigl[\text{for every }\chi\in\pmone^n,\ \Fnorm{F_\chi}\ge c_A(\log_2n)^{3/2}\Bigr]\ge1-e^{-An}.
\]
\end{theorem}

The parameter $A$ allows any fixed exponential failure rate, at the cost of a constant $c_A$ that depends on $A$. We do not optimize $c_A$; it is very small, so the bound is of interest only for large $n$. Since $1-e^{-An}>0$ and the points are almost surely distinct, there is an $n$-point set for which every coloring leaves an anchored rectangle with imbalance at least $c_A(\log_2n)^{3/2}$. Anchored rectangles are a subfamily of all rectangles, so together with Nikolov's upper bound~\cite{Nik17} we obtain:

\begin{corollary}[Tusn\'ady's problem in the plane]\label{cor:main}
The worst-case combinatorial discrepancy of $n$ points in the plane with respect to axis-parallel rectangles is $\Theta(\log^{3/2}n)$.
\end{corollary}

\paragraph{Previous bounds.}
Tusn\'ady asked whether $\Delta_2(n)$ is bounded by a constant. Beck~\cite{Bec81} answered this question negatively and gave an upper bound of $O(\log^4n)$. The lower bound $\Delta_2(n)=\Omega(\log n)$ follows from Schmidt's theorem on continuous discrepancy~\cite{Sch72} by transference, as used by Beck~\cite{Bec81}; see also \cite{Mat99,ABN18}. A sequence of improvements (\cref{tab:bounds}) led to the upper bound of $O(\log^{3/2}n)$, which Nikolov~\cite{Nik17} proved using factorization norms and Banaszczyk's theorem on signed series~\cite{Ban12}. This left a factor of $\sqrt{\log n}$ between the bounds, and Bansal~\cite{Ban22} conjectured that the lower bound is tight. \Cref{thm:main} closes this gap in favor of the upper bound: random point sets attain Nikolov's bound up to a constant factor.

\begin{table}[H]
\centering
\begin{tabular}{@{}lll@{}}
\toprule
bound on $\Delta_2(n)$ & reference & remark\\
\midrule
$O(\log^4n)$ & Beck~\cite{Bec81} & \\
$O(\log^{3.5+\eps}n)$ & Beck~\cite{Bec89} & every fixed $\eps>0$\\
$O(\log^3n)$ & Bohus~\cite{Boh90} & \\
$O(\log^{2.5}n)$ & Srinivasan~\cite{Sri97} & \\
$O(\log^2n)$ & Bansal and Garg~\cite{BG17} & constructive\\
$O(\log^{3/2}n)$ & Nikolov~\cite{Nik17} & \\
\midrule
$\Omega(\log n)$ & Schmidt~\cite{Sch72} & via transference \cite{Bec81,Mat99,ABN18}\\
$\Omega(\log n)$ & Matou\v{s}ek and Nikolov~\cite{MN15} & factorization norms\\
$\Omega(\log^{3/2}n)$ & \textbf{this paper} & random point sets\\
\bottomrule
\end{tabular}
\caption{Bounds for the combinatorial discrepancy of rectangles in the plane.}
\label{tab:bounds}
\end{table}

\paragraph{The proof in four steps.}
Fix a coloring $\chi$ and a large constant $b$, and write $F=F_\chi$. We reveal the horizontal coordinates of the points one base-$b$ digit at a time, while keeping all vertical coordinates visible. Let $\F_j$ be the information after $j$ digits, that is, all vertical coordinates and the first $j$ digits of every horizontal coordinate. We reveal $h=\Theta_b(\log n)$ digits, and call the passage from $\F_j$ to $\F_{j+1}$ a \emph{transition}.
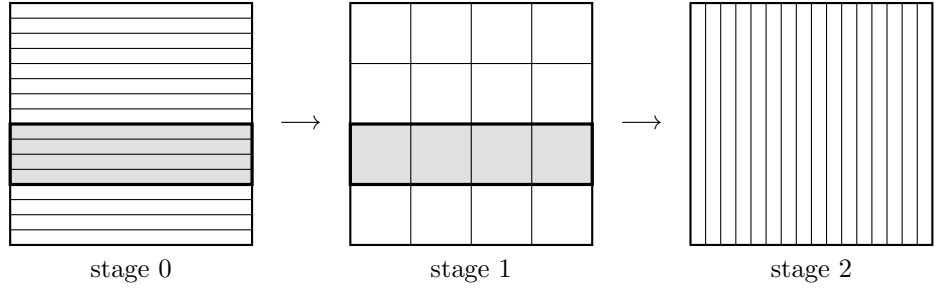
\begin{figure}[!t]
\centering
\begin{tikzpicture}[x=3.2cm,y=3.2cm,font=\small]
  \begin{scope}
    \fill[black!12] (0,0.25) rectangle (1,0.5);
    \draw[thick] (0,0) rectangle (1,1);
    \foreach \t in {1,...,15} \draw (0,\t/16) -- (1,\t/16);
    \draw[very thick] (0,0.25) rectangle (1,0.5);
    \node[below] at (0.5,-0.02) {stage $0$};
  \end{scope}
  \node at (1.2,0.5) {$\longrightarrow$};
  \begin{scope}[xshift=4.5cm]
    \fill[black!12] (0,0.25) rectangle (1,0.5);
    \draw[thick] (0,0) rectangle (1,1);
    \foreach \t in {1,2,3} { \draw (0,\t/4) -- (1,\t/4); \draw (\t/4,0) -- (\t/4,1); }
    \draw[very thick] (0,0.25) rectangle (1,0.5);
    \node[below] at (0.5,-0.02) {stage $1$};
  \end{scope}
  \node at (2.61,0.5) {$\longrightarrow$};
  \begin{scope}[xshift=9cm]
    \draw[thick] (0,0) rectangle (1,1);
    \foreach \t in {1,...,15} \draw (\t/16,0) -- (\t/16,1);
    \node[below] at (0.5,-0.02) {stage $2$};
  \end{scope}
\end{tikzpicture}
\caption{The partitions, drawn for $b=4$ and $h=2$. Every stage has $b^h=16$ cells, which become narrower and taller as more digits are revealed. The shaded region is one of the four comparison rectangles of the transition $0\to1$: a union of $b$ cells of stage~$0$ and also of $b$ cells of stage~$1$.}
\label{fig:stages}
\end{figure}
\begin{enumerate}[label=(\arabic*),itemsep=2pt]
\item \emph{Potential.} At stage $j$, that is, after $j$ digits have been revealed, the unit square is partitioned into cells of width $b^{-j}$ and height $b^{j-h}$ (\cref{fig:stages}). For a horizontal interval $I=[I^-,I^+]$, let $a^{(I)}(y)=\frac12\bigl(F(I^-,y)+F(I^+,y)\bigr)$ be the mean of the two values of $F$ at its endpoints, as a function of the height $y$. The \emph{oscillation} of a function $f$ on a vertical interval $B$ is $\osc_Bf=\max_Bf-\min_Bf$. The potential is the average oscillation of the $b^h$ cells:
\begin{equation}\label{eq:Zdef}
Z_j=\frac1{b^h}\sum_{\text{stage-$j$ cells }I\times B}\osc_Ba^{(I)}.
\end{equation}
It is determined by $\F_j$, and it is at most $2\Fnorm F$. So it suffices to show that $Z_h$ is large. The \emph{conditional gain} of a transition is $\E[Z_{j+1}-Z_j\mid\F_j]$. Steps~(2) and~(3) show that the conditional gains of the $h$ transitions add up to a large number, and step~(4) shows that $Z_h$ does not fall far below their sum.
\item \emph{Local gain} (\cref{lem:local,lem:drift}). A comparison rectangle is a union of $b$ old cells, from before a transition, and also of $b$ new cells, from after it. In conditional expectation given $\F_j$, the oscillations of the new cells add up to at least those of the old cells plus $\frac14\sqrt{m_R}$, where $m_R$ is the number of points in the \emph{inner region}, that is, the rectangle without its bottom and top old cells. The gain comes from the newly revealed digits of these $m_R$ points. After averaging over the $b$ new cells, they shift the values on the top old cell against those on the bottom old cell by a symmetric random amount of order $\sqrt{m_R}$, and only the new cells are high enough to register this shift. This comparison uses the opposing coordinate scales of Talagrand's induction \cite[Section~3]{Tal94}.
\item \emph{Good transitions} (\cref{lem:crowd}). The bound of step~(2) is weak if the points crowd into few inner regions, since $\sqrt{m_R}$ grows slowly. Call an inner region \emph{heavy} if it holds more than $b$ times the average number of points of a comparison rectangle, and call a transition \emph{good} if at least half of the points lie in inner regions that are not heavy. If a transition is not good, more than half of the points lie in a set of area at most $3/b$. When a transition has at most $n$ comparison rectangles, a union bound shows that this is exponentially unlikely. So, with overwhelming probability, all $h$ transitions are good. This event depends only on the points, so it is common to every coloring.
\item \emph{All colorings} (\cref{lem:fluct,prop:sim}). At each transition, the newly revealed digit of one point changes the new potential very little. Hence, by bounded differences, $Z_h$ is unlikely to fall far below the sum of the conditional gains, and for large $b$ a union bound over the $2^n$ fixed colorings gives a statement for all colorings, including those chosen after observing the point set.
\end{enumerate}
We choose $h$ so that a comparison rectangle holds $\Theta_b(h)$ points on average. Then the conditional gain of a good transition is $\Omega_b(\sqrt h)$, and on the event of step~(3) these conditional gains add up to $\Omega_b(h^{3/2})$. A denser choice, with more points per comparison rectangle, would increase each gain. But a transition has $M=b^{h-1}$ comparison rectangles, and for a deviation of the order of the total gain, the fluctuation bound of step~(4) has an exponent of order $hM$ for fixed $b$, whatever the density $n/M$ is. This exponent must exceed $n\log2$ to pay for the union bound. So a density $n/M$ of order $h$ is the largest that the argument allows.

\begin{table}[!b]
\centering
\small
\renewcommand{\arraystretch}{1.2}
\begin{tabular}{@{}lll@{}}
\toprule
Symbol & Meaning & Defined in\\
\midrule
$R=I\times J$ & a rectangle; from \cref{sec:scales} on, a comparison rectangle & \cref{sec:setup}\\
$B_1,\dots,B_b$ & blocks of $J$, from bottom to top & \cref{sec:blocks}\\
$C_1,\dots,C_b$ & children of $I$, from left to right & \cref{sec:onerect}\\
$\Rin$, $m_R$ & inner region of $R$, and the number of points in it & \cref{sec:blocks}\\
\addlinespace[5pt]
$F=F_\chi$ & signed count of the points in $[0,x]\times[0,y]$ & Section~1\\
$\osc_Bf$, $\midp_Bf$ & oscillation and midpoint of $f$ on the vertical interval $B$ & \eqref{eq:osc}\\
$a^{(I)}$; $a$ & endpoint average of $I$; written $a$ while $I$ is fixed (\cref{sec:setup,sec:blocks,sec:onerect}) & \eqref{eq:a}\\
$a^{(C_k)}$; $a_k$ & endpoint average of the child $C_k$; written $a_k$ while $I$ is fixed & \cref{sec:onerect}\\
$\bar a$ & average of the child functions $a_1,\dots,a_b$ & \cref{sec:onerect}\\
$X_i$, $D$ & jump of $\bar a-a$ at the height of $P_i$; sum of these jumps in $\Rin$ & \cref{sec:onerect}\\
\addlinespace[5pt]
$h$, $M=b^{h-1}$ & number of transitions; number of comparison rectangles & \cref{sec:scales}\\
$Z_j$, $\F_j$ & potential at stage $j$; information after $j$ digits & \cref{sec:scales}\\
$T$ & fluctuation term & \eqref{eq:roadmap}\\
$\delta_j$ & lower bound for the conditional gain & \eqref{eq:drift}\\
$\delta_*$ & lower bound for $\delta_j$ at a good transition & \eqref{eq:d}\\
$G_j$, $G$ & the transition $j\to j+1$ is good; all transitions are good & \cref{sec:scales}\\
\bottomrule
\end{tabular}
\caption{Notation of \cref{sec:proof}.}
\label{tab:notation}
\end{table}

\section{The Proof}\label{sec:proof}

\subsection{Setup}\label{sec:setup}

Throughout this section, $P_1,\dots,P_n$ are independent uniform points in $[0,1]^2$. Fix a coloring $\chi$ and write $F=F_\chi$. Let $b\ge3$ be an integer; it is the base in which we write the horizontal coordinates. We call $V_i$ the \emph{height} of $P_i$, and we say that $P_i$ lies in a horizontal interval $I$ if $U_i\in I$. We work on the probability-one event that all heights are distinct and that no coordinate of a point has the form $k/b^r$ with integers $k$ and $r\ge0$. All intervals of \cref{sec:scales} have endpoints of this form. \Cref{tab:notation} lists the notation.

Until the end of \cref{sec:onerect}, we fix one rectangle $R=I\times J$ (\cref{fig:local}), with a horizontal interval $I=[I^-,I^+]$ and a vertical interval $J$. In \cref{sec:scales}, it will be a comparison rectangle of \cref{fig:stages}.

\paragraph{Endpoint average.}
The \emph{endpoint average} of $I$ is the function of the height
\begin{equation}\label{eq:a}
a(y)=\frac{F(I^-,y)+F(I^+,y)}2=\sum_i\chi_i\,w_i\,\1\{V_i\le y\},
\end{equation}
where the \emph{weight} $w_i$ of $P_i$ in $I$ is $1$ if $P_i$ lies to the left of $I$, $1/2$ if it lies in $I$, and $0$ if it lies to the right of $I$. Indeed, among the points of height at most $y$, $F(I^-,y)$ counts those to the left of $I$, each with its color, and $F(I^+,y)$ counts these points and also those in $I$. The endpoints of $I$ do not matter, since almost surely no point lies on an endpoint of the intervals we use. So $a$ depends on the horizontal coordinates only through the position of each point relative to $I$. It is the function $a^{(I)}$ of \cref{sec:intro}; while $I$ is fixed, we omit the superscript. The weight $1/2$ of a point in $I$ is what makes the averaging step of \cref{sec:onerect} unbiased; see \eqref{eq:avg}.

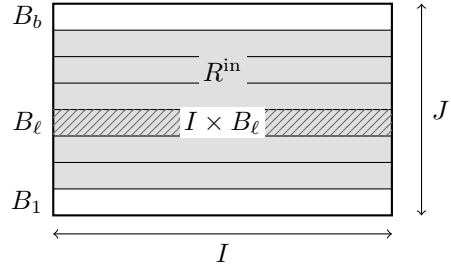
\begin{figure}[!htb]
\centering
\begin{tikzpicture}[x=0.56cm,y=0.35cm,font=\small]
  \fill[black!12] (0,1) rectangle (8,7);
  \draw[thick] (0,0) rectangle (8,8);
  \foreach \t in {1,...,7} \draw (0,\t) -- (8,\t);
  \node[left] at (0,0.5) {$B_1$};
  \node[left] at (0,7.5) {$B_b$};
  \node[left] at (0,3.5) {$B_\ell$};
  \fill[pattern=north east lines,pattern color=black!55] (0,3) rectangle (8,4);
  \node[fill=white,inner sep=1.5pt] at (4,3.5) {$I\times B_\ell$};
  \node[fill=black!12,inner sep=1.5pt] at (4,5.5) {$\Rin$};
  \draw[<->] (0,-0.7) -- node[below] {$I$} (8,-0.7);
  \draw[<->] (8.7,0) -- node[right] {$J$} (8.7,8);
\end{tikzpicture}
\caption{The rectangle $R=I\times J$ and its $b$ horizontal strips, drawn for $b=8$. The inner region $\Rin$ (shaded) is $R$ without its bottom and top strips.}
\label{fig:local}
\end{figure}

\paragraph{Oscillation and midpoint.}
For a function $f$ of the height and a vertical interval $B$, the \emph{oscillation} and the \emph{midpoint} of $f$ on $B$ are
\begin{equation}\label{eq:osc}
\osc_Bf=\max_Bf-\min_Bf,\qquad\midp_Bf=\frac{\max_Bf+\min_Bf}2,
\end{equation}
where $\max_Bf$ and $\min_Bf$ are the extrema of $f(y)$ over $y\in B$. Vertical intervals are closed at the bottom and open at the top. All functions of the height that we consider are finite linear combinations of the functions $y\mapsto\1\{V_i\le y\}$, up to an additive constant. They take finitely many values, so these extrema exist, and they are \mbox{$\midp_Bf\pm\frac12\osc_Bf$}. Clearly $\osc_B(cf)=|c|\osc_Bf$ for real $c$. The oscillation is \emph{shift invariant}: adding a constant $c$ to $f$ does not change $\osc_Bf$, whereas it adds $c$ to $\midp_Bf$. The oscillation also satisfies the triangle inequality
\[
\osc_B(f+g)\le\osc_Bf+\osc_Bg,
\]
because the maximum of $f+g$ on $B$ is at most $\max_Bf+\max_Bg$, and its minimum is at least $\min_Bf+\min_Bg$. By induction, it extends to finite sums. Applying it to $f=(f+g)+(-g)$ and using $\osc_B(-g)=\osc_Bg$ gives $\osc_Bf\le\osc_B(f+g)+\osc_Bg$. Together with the triangle inequality itself, this is the reverse triangle inequality
\[
|\osc_B(f+g)-\osc_Bf|\le\osc_Bg.
\]
Shift invariance is the reason why we work with oscillations and not with maxima. When one point moves horizontally, its effect on a vertical interval that does not contain its height is a constant shift, and the oscillation does not see it.

\subsection{A deterministic comparison}\label{sec:blocks}

Partition $J$ into blocks $B_1,\dots,B_b$ of equal length, from bottom to top. The rectangle $R$ consists of the \emph{horizontal strips} $I\times B_\ell$ (\cref{fig:local}). Set
\[
\Rin=I\times(B_2\cup\dots\cup B_{b-1}),
\]
and call $\Rin$ the \emph{inner region} of $R$. It is $R$ without its bottom and top strips, and it is not empty because $b\ge3$. We write $m_R$ for the number of points in $\Rin$.

The first lemma compares the oscillation on $J$ with the oscillations on the blocks.

\begin{lemma}[Vertical comparison]\label{lem:comparison}
Let $f$ be a function on $J$ with finitely many values. Then
\begin{equation}\label{eq:det}
b\,\osc_Jf\ \ge\ \sum_{\ell=1}^b\osc_{B_\ell}f+2\,\bigl|\midp_{B_b}f-\midp_{B_1}f\bigr|.
\end{equation}
\end{lemma}

We call $\midp_{B_b}f-\midp_{B_1}f$ the \emph{offset} between the top block and the bottom block; the last term of \eqref{eq:det} is twice its absolute value. The offset is the source of the gain: adding a constant to $f$ on the top block changes the offset but none of the oscillations on the blocks, so among the terms of \eqref{eq:det} only the oscillation on $J$ can register it.

\begin{proof}
Since $J$ contains $B_1$ and $B_b$, and the extrema of $f$ on a block $B$ are \mbox{$\midp_Bf\pm\frac12\osc_Bf$},
\begin{align*}
\osc_Jf&\ \ge\ \max_{B_b}f-\min_{B_1}f=\frac{\osc_{B_1}f+\osc_{B_b}f}2+\bigl(\midp_{B_b}f-\midp_{B_1}f\bigr),\\
\osc_Jf&\ \ge\ \max_{B_1}f-\min_{B_b}f=\frac{\osc_{B_1}f+\osc_{B_b}f}2-\bigl(\midp_{B_b}f-\midp_{B_1}f\bigr).
\end{align*}
Taking the larger of the two right-hand sides,
\[
\osc_Jf\ \ge\ \frac{\osc_{B_1}f+\osc_{B_b}f}2+\bigl|\midp_{B_b}f-\midp_{B_1}f\bigr|.
\]
Also $\osc_Jf\ge\osc_{B_\ell}f$ for each of the remaining $b-2$ blocks $B_\ell$. Adding these $b-2$ inequalities and twice the last display proves \eqref{eq:det}.
\end{proof}

\subsection{Horizontal refinement and the random gain}\label{sec:onerect}

Now partition $I$ into children $C_1,\dots,C_b$ of equal length, from left to right. The same rectangle $R$ also consists of the \emph{vertical strips} $C_k\times J$ (\cref{fig:gain}, left). Let $a_k$ be the endpoint average of $C_k$, defined as in \eqref{eq:a} with $C_k$ in place of $I$. The word \emph{old} always refers to $I$ before this refinement: we call $a$ the \emph{old function} and $a_1,\dots,a_b$ the \emph{child functions}.

\emph{Local experiment:} fix all colors, all heights, and the position of every point relative to $I$, that is, whether it lies to the left of $I$, in $I$, or to the right of $I$. Each point in $I$ falls into one of the children $C_1,\dots,C_b$, uniformly at random and independently of the other points. In this experiment the old function is fixed, and the child functions are random. Nothing else needs to be specified: a point outside $I$ lies on the same side of every child, so the child functions depend on the horizontal coordinates only through the positions relative to $I$ and the children of the points in $I$.

The gain will come from a sum of independent \emph{symmetric} random variables, that is, variables that have the same distribution as their negatives. The following fact bounds the expected absolute value of such a sum from below, even after a shift by an arbitrary number $\theta$. For uniform signs it is Khintchine's inequality with a weaker constant; the fourth-moment hypothesis holds for uniform signs, and we will check it for the variables that we need. We prove the fact in \cref{app:sym}.

\begin{fact}[Sums of symmetric variables]\label{fact:sym}
Let $X_1,\dots,X_m$ be independent symmetric random variables. Suppose that $\E X_i^2=\sigma^2$ and $\E X_i^4\le3\sigma^4$ for every $i$, where $\sigma>0$. Then
\[
\E|\theta+X_1+\dots+X_m|\ \ge\ \sigma\sqrt{m/3}\qquad\text{for every real }\theta.
\]
\end{fact}

The next lemma is the heart of the proof. Cutting $R$ into vertical strips instead of horizontal ones increases the total oscillation in expectation, by at least $\frac14\sqrt{m_R}$, where $m_R$ is the number of points in the inner region.

\begin{lemma}[Local gain]\label{lem:local}
In the local experiment,
\[
\E\sum_{k=1}^b\osc_Ja_k\ \ge\ \sum_{\ell=1}^b\osc_{B_\ell}a+\frac14\sqrt{m_R}.
\]
\end{lemma}

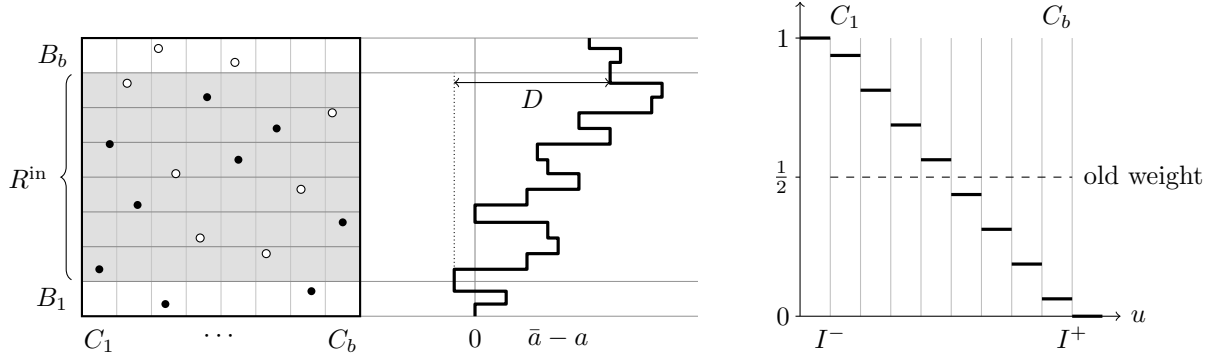
\begin{figure}[t]
\centering
\begin{tikzpicture}[font=\small]
  \begin{scope}[x=1cm,y=0.46cm]
    \foreach \t in {0,1,7,8} \draw[black!40] (3.68,\t) -- (8.15,\t);
  \end{scope}
  \begin{scope}[x=0.46cm,y=0.46cm]
    \fill[black!12] (0,1) rectangle (8,7);
    \foreach \t in {1,...,7} \draw[black!22] (\t,0) -- (\t,8);
    \foreach \t in {1,...,7} \draw[black!45] (0,\t) -- (8,\t);
    \draw[thick] (0,0) rectangle (8,8);
    \foreach \x/\y in {2.4/0.35,6.6/0.72,0.5/1.35,7.5/2.7,1.6/3.2,4.5/4.5,0.8/4.95,5.6/5.4,3.6/6.3}
      \fill (\x,\y) circle (1.5pt);
    \foreach \x/\y in {5.3/1.8,3.4/2.25,6.3/3.65,2.7/4.1,7.2/5.85,1.3/6.7,4.4/7.3,2.2/7.7}
      \draw[fill=white] (\x,\y) circle (1.5pt);
    \node[left] at (-0.1,0.5) {$B_1$};
    \node[left] at (-0.1,7.5) {$B_b$};
    \draw[decorate,decoration={brace,amplitude=4pt}] (-0.3,1.05) -- (-0.3,6.95) node[midway,left=5pt] {$\Rin$};
    \node[below] at (0.5,-0.1) {$C_1$};
    \node[below] at (4,-0.1) {$\cdots$};
    \node[below] at (7.5,-0.1) {$C_b$};
  \end{scope}
  \begin{scope}[xshift=5.2cm,x=2.2cm,y=0.46cm]
    \draw[black!45] (0,0) -- (0,8);
    \node[below] at (0,-0.1) {$0$};
    \draw[densely dotted] (-0.125,1) -- (-0.125,7);
    \draw[<->] (-0.125,6.72) -- (0.8125,6.72) node[midway,below=-1pt] {$D$};
    \draw[very thick,line join=round] (0,0) -- (0,0.35) -- (0.1875,0.35) -- (0.1875,0.72) -- (-0.125,0.72) -- (-0.125,1.35) -- (0.3125,1.35) -- (0.3125,1.8) -- (0.5,1.8) -- (0.5,2.25) -- (0.4375,2.25) -- (0.4375,2.7) -- (0,2.7) -- (0,3.2) -- (0.3125,3.2) -- (0.3125,3.65) -- (0.625,3.65) -- (0.625,4.1) -- (0.4375,4.1) -- (0.4375,4.5) -- (0.375,4.5) -- (0.375,4.95) -- (0.8125,4.95) -- (0.8125,5.4) -- (0.625,5.4) -- (0.625,5.85) -- (1.0625,5.85) -- (1.0625,6.3) -- (1.125,6.3) -- (1.125,6.7) -- (0.8125,6.7) -- (0.8125,7.3) -- (0.875,7.3) -- (0.875,7.7) -- (0.6875,7.7) -- (0.6875,8);
    \node[below] at (0.5,-0.1) {$\bar a-a$};
  \end{scope}
  \begin{scope}[xshift=9.5cm,x=0.40cm,y=3.68cm]
    \foreach \r in {1,...,9} \draw[black!25] (\r,0) -- (\r,1);
    \draw[->] (0,0) -- (10.6,0) node[right] {$u$};
    \draw[->] (0,0) -- (0,1.13);
    \foreach \v/\l in {0/0,0.5/\frac12,1/1}
      \draw (-0.15,\v) -- (0,\v) node[left] {$\l$};
    \draw[dashed] (1,0.5) -- (9,0.5);
    \draw[very thick] (0,1) -- (1,1);
    \foreach \r in {1,...,8} \draw[very thick] (\r,{1-(\r-0.5)/8}) -- (\r+1,{1-(\r-0.5)/8});
    \draw[very thick] (9,0) -- (10,0);
    \node[right] at (9.05,0.5) {old weight};
    \node[below] at (1,0) {$I^-$};
    \node[below] at (9,0) {$I^+$};
    \node[above] at (1.5,1) {$C_1$};
    \node[above] at (8.5,1) {$C_b$};
  \end{scope}
\end{tikzpicture}
\caption{The gain, drawn for $b=8$. Left: the points of $R$; filled points have color $+1$ and hollow points have color $-1$. The inner region $\Rin$ is shaded, and the columns are the vertical strips $C_k\times J$. Middle: the difference $\bar a-a$ on $J$, with its values plotted horizontally from the thin vertical line at $0$; here $\bar a$ is the average of the child functions. By \eqref{eq:walk}, this difference jumps by $X_i$ at the height of each point of $R$. The jumps of the points of $\Rin$ add up to $D$, which shifts the top block against the bottom block. Right: the weight of a point in $\bar a$ as a function of its horizontal coordinate $u$. In the child $C_r$ it is $\frac12+\frac{b+1-2r}{2b}$, and the dashed line is its old weight $1/2$, that is, its weight in $a$.}
\label{fig:gain}
\end{figure}

\emph{Why the lemma holds.} We will see that the average $\bar a$ of the $b$ child functions is the old function plus an independent jump at the height of each point in $I$ (\cref{fig:gain}, middle). Each jump is symmetric. The points of $\Rin$ lie above the bottom block $B_1$ and below the top block $B_b$. So the sum $D$ of their jumps shifts $\bar a$ on $B_b$ against $\bar a$ on $B_1$, and $D$ has standard deviation of order $\sqrt{m_R}$. By \cref{lem:comparison}, the oscillation on $J$ registers this offset. The oscillation on a single block, which is all that a horizontal strip measures, does not. For a particular outcome, $D$ may cancel an offset that is already there. On average it cannot, because $D$ is symmetric; this is what \cref{fact:sym} says.

\begin{proof}[Proof of \cref{lem:local}]
\emph{Step 1: averaging the children.} Average the child functions:
\[
\bar a(y)=\frac1b\sum_{k=1}^ba_k(y).
\]
By the triangle inequality for the oscillation, $\sum_k\osc_Ja_k\ge\osc_J\bigl(\sum_ka_k\bigr)=b\,\osc_J\bar a$. So \cref{lem:comparison}, applied to $f=\bar a$, gives
\begin{equation}\label{eq:avgosc}
\sum_k\osc_Ja_k\ \ge\ \sum_\ell\osc_{B_\ell}\bar a+2\,\bigl|\midp_{B_b}\bar a-\midp_{B_1}\bar a\bigr|.
\end{equation}

\emph{Step 2: the weights in $\bar a$.} By \eqref{eq:a}, the function $\bar a$ counts each point with the average of its weights in the $b$ children, whereas $a$ counts it with its weight in $I$, which we call its \emph{old weight}. A point outside $I$ lies to the left of every child or to the right of every child, so it keeps its old weight, which is $1$ or~$0$. A point in $I$ has old weight $1/2$. If it is in the child $C_r$, it lies to the left of the $b-r$ children $C_{r+1},\dots,C_b$, where it has weight $1$, and inside $C_r$, where it has weight $1/2$. So its weight in $\bar a$ is
\[
\frac{(b-r)+1/2}b=\frac12+\frac{b+1-2r}{2b},
\]
that is, its old weight $1/2$ plus a deviation; see \cref{fig:gain} (right). For a point $P_i$ in the child $C_r$, let $X_i=\chi_i\,(b+1-2r)/(2b)$ be this deviation, multiplied by the color. The old function $a$ counts every point with its old weight, by \eqref{eq:a}. So, for every height $y$,
\begin{equation}\label{eq:walk}
\bar a(y)=a(y)+\sum_{i:\,U_i\in I,\ V_i\le y}X_i.
\end{equation}
In words, $\bar a$ follows the old function and makes an extra jump $X_i$ at the height of each point in $I$.

In the local experiment, the index $r$ of a point in $I$ is uniform on $\{1,\dots,b\}$, independently for different points. So $b+1-2r$ is uniform on the symmetric set $\{1-b,\,3-b,\,\dots,\,b-1\}$, and the $X_i$ are independent and symmetric, whatever the colors are. In particular $\E X_i=0$, and \eqref{eq:walk} gives
\begin{equation}\label{eq:avg}
\E\,\bar a(y)=a(y)\qquad\text{for every height }y.
\end{equation}

\emph{Step 3: taking expectations.} The expectation of a maximum is at least the maximum of the expectations, and the expectation of a minimum is at most the minimum of the expectations. So \eqref{eq:avg} gives, for every $\ell$,
\[
\E\,\osc_{B_\ell}\bar a=\E\max_{B_\ell}\bar a-\E\min_{B_\ell}\bar a\ \ge\ \max_{B_\ell}a-\min_{B_\ell}a=\osc_{B_\ell}a.
\]
Taking expectations in \eqref{eq:avgosc}, we get
\begin{equation}\label{eq:gain}
\E\sum_k\osc_Ja_k\ \ge\ \sum_\ell\osc_{B_\ell}a+2\,\E\bigl|\midp_{B_b}\bar a-\midp_{B_1}\bar a\bigr|.
\end{equation}
The first term on the right is what the horizontal strips had. The second term is the gain.

\emph{Step 4: the gain.} Let $D=\sum_{i:\,P_i\in\Rin}X_i$ be the sum of the jumps of the $m_R$ points of $\Rin$, and let $g$ be the function $\bar a$ without these jumps:
\[
g(y)=a(y)+\sum_{i:\,U_i\in I,\ P_i\notin\Rin,\ V_i\le y}X_i.
\]
The heights of the points of $\Rin$ lie above $B_1$ and below $B_b$. So by \eqref{eq:walk}, their jumps do not enter $\bar a$ on $B_1$, and they add the constant $D$ to $\bar a$ on $B_b$; see \cref{fig:gain} (left and middle). That is, $\bar a=g$ on $B_1$ and $\bar a=g+D$ on $B_b$. Adding the constant $D$ to a function adds $D$ to its midpoint, so
\[
\midp_{B_b}\bar a-\midp_{B_1}\bar a=\theta+D,\qquad\text{where }\theta=\midp_{B_b}g-\midp_{B_1}g.
\]

Now condition on the children of all other points in $I$, that is, of the points in $I$ whose heights lie in $B_1$, in $B_b$, or outside $J$. This fixes the function $g$, and hence the number $\theta$. The children of the $m_R$ points of $\Rin$ are independent of the conditioning, so the jumps in $D$ are still independent and symmetric. By definition, $X_i^2=(r-\frac{b+1}2)^2/b^2$ with $r$ uniform on $\{1,\dots,b\}$. Here $r$ has mean $(b+1)/2$ and second moment $(b+1)(2b+1)/6$, so its variance is $(b+1)(2b+1)/6-(b+1)^2/4=(b^2-1)/12$. Also $|X_i|\le(b-1)/(2b)$. Hence
\[
\E X_i^2=\frac{b^2-1}{12b^2}=:\sigma^2,\qquad\E X_i^4\le\Bigl(\frac{b-1}{2b}\Bigr)^2\sigma^2\le\frac{b^2-1}{4b^2}\,\sigma^2=3\sigma^4.
\]
So \cref{fact:sym} applies to $\theta+D$, with $m=m_R$. Its bound holds for every $\theta$, so the children of the other points cannot reduce it. Conditionally, and hence also without the conditioning,
\[
\E\bigl|\midp_{B_b}\bar a-\midp_{B_1}\bar a\bigr|\ \ge\ \sigma\sqrt{m_R/3}=\frac{\sqrt{b^2-1}}{6b}\,\sqrt{m_R}\ \ge\ \frac{\sqrt{m_R}}8,
\]
because $16(b^2-1)\ge9b^2$ for $b\ge2$. Together with \eqref{eq:gain}, this proves the inequality of the lemma.
\end{proof}

\paragraph{What one rectangle gives, and what is missing.}
Take $I=[0,1]$ and $J=[0,1)$, so that $R$ is the unit square without its top edge. Every oscillation of a child function is at most $2\Fnorm F$, because it is a difference of two values of $a_k$, and each of them is an average of two values of $F$. So the left-hand side in \cref{lem:local} is at most $2b\,\E\Fnorm F$. When the heights are fixed and the horizontal coordinates are independent and uniform, the first digits perform the local experiment, and the lemma gives $\E\Fnorm F\ge\sqrt{m_R}/(8b)$. For typical heights, $\Rin$ contains about $(1-2/b)n$ points, so this is of order $\sqrt n$ for fixed $b$. But it is not a discrepancy lower bound. The coloring was fixed before the horizontal coordinates were sampled, and a coloring chosen after seeing the points achieves $O(\log^{3/2}n)$~\cite{Nik17}. We need a bound for all $2^n$ colorings simultaneously, so each fixed coloring may fail only with probability well below $2^{-n}$. One comparison on one rectangle is far from that. \Cref{sec:scales} therefore uses $M$ small rectangles in place of one large rectangle, and repeats the comparison at $h$ successive horizontal scales. The lemma is designed for this: its output, a sum of oscillations over vertical strips, has the same form as its input, a sum of oscillations over horizontal strips. Two things are then still missing. The gain depends on how many points each rectangle receives, which is the subject of \cref{lem:crowd}, and the lemma holds only in expectation, which is the subject of \cref{lem:fluct}.

\subsection{All scales}\label{sec:scales}

\paragraph{Partitions and potential.}
Fix an integer $h\ge1$ and set $M=b^{h-1}$. At stage $j$, $0\le j\le h$, cut $[0,1]$ horizontally into $b^j$ intervals of length $b^{-j}$, and vertically into $b^{h-j}$ intervals of length $b^{j-h}$. Their products are the cells of stage $j$. There are $b^j\cdot b^{h-j}=bM$ cells at every stage. The interval $I$ is no longer fixed, and $a^{(I)}$ denotes the endpoint average of $I$. As in \eqref{eq:Zdef}, the potential is
\[
Z_j=\frac1{bM}\sum_{\text{stage-$j$ cells }I\times B}\osc_Ba^{(I)}.
\]
By \eqref{eq:osc}, an oscillation $\osc_Ba^{(I)}$ is a difference of two values of $a^{(I)}$, and each of them is an average of two values of $F$. So every oscillation is at most $2\Fnorm F$, and
\begin{equation}\label{eq:Zbound}
0\le Z_j\le2\Fnorm F.
\end{equation}
Let $\F_j$ be the $\sigma$-field generated by all heights and by the first $j$ base-$b$ digits of every horizontal coordinate. The first $j$ digits of $U_i$ determine the stage-$j$ horizontal interval that contains $U_i$. So $Z_j$ is $\F_j$-measurable. The $(j+1)$-st digits, which we call the \emph{next digits} at stage $j$, are independent and uniform conditional on $\F_j$. The passage from stage $j$ to stage $j+1$ is the \emph{transition} $j\to j+1$.

\paragraph{The plan.}
Since $Z_j$ is $\F_j$-measurable, the increment $Z_{j+1}-Z_j$ is its conditional mean given $\F_j$ plus $Z_{j+1}-\E[Z_{j+1}\mid\F_j]$. Summing over $j$ gives the identity
\begin{equation}\label{eq:roadmap}
Z_h\ =\ Z_0\ +\ \sum_{j<h}\E[Z_{j+1}-Z_j\mid\F_j]\ +\ T,\qquad\text{where}\quad T=\sum_{j<h}\bigl(Z_{j+1}-\E[Z_{j+1}\mid\F_j]\bigr).
\end{equation}
We call $\E[Z_{j+1}-Z_j\mid\F_j]$ the \emph{conditional gain} of the transition $j\to j+1$, and $T$ the \emph{fluctuation term}. By \eqref{eq:Zbound}, $Z_0\ge0$ and $Z_h\le2\Fnorm F$. The rest of the proof has three steps.
\begin{enumerate}[label=(\alph*),nosep]
\item Each conditional gain is at least a number $\delta_j$ that depends only on the points (\cref{lem:drift}).
\item When $M\le n$ and $b$ is large, $\delta_j$ is large for every $j$, except with exponentially small probability (\cref{lem:crowd}).
\item For a fixed coloring, the fluctuation term $T$ is unlikely to be very negative (\cref{lem:fluct}).
\end{enumerate}
\Cref{prop:sim} puts the three steps together and takes a union bound over the colorings.

\paragraph{Transitions.}
In the transition $j\to j+1$, we call the objects of stage $j$ \emph{old} and those of stage $j+1$ \emph{new}. A \emph{comparison rectangle} $R=I\times J$ is the product of one old horizontal interval and one new vertical interval. It is the union of $b$ old cells, which are its horizontal strips, and also of $b$ new cells, which are its vertical strips, as in \cref{fig:stages}. There are $b^j$ old horizontal intervals and $b^{h-j-1}$ new vertical intervals, hence $M$ comparison rectangles, each of area $b^{-j}\cdot b^{j+1-h}=1/M$. As in \cref{sec:blocks}, the inner region of $R$ is $R$ without its bottom and top old cells, and $m_R$ is the number of points in it.

Summing \cref{lem:local} over the comparison rectangles bounds the conditional gain of the transition.

\begin{lemma}[Conditional gain]\label{lem:drift}
For every $0\le j<h$, almost surely
\begin{equation}\label{eq:drift}
\E[Z_{j+1}-Z_j\mid\F_j]\ \ge\ \delta_j:=\frac1{4bM}\sum_R\sqrt{m_R},
\end{equation}
where $R$ ranges over the $M$ comparison rectangles of the transition $j\to j+1$.
\end{lemma}

\begin{proof}
Every old cell is a horizontal strip of exactly one comparison rectangle $R=I\times J$, and every new cell is a vertical strip of exactly one. So, with the children $C_1,\dots,C_b$ of $I$ and the blocks $B_1,\dots,B_b$ of $J$ for each $R$,
\[
bM\,(Z_{j+1}-Z_j)=\sum_R\Bigl(\sum_k\osc_Ja^{(C_k)}-\sum_\ell\osc_{B_\ell}a^{(I)}\Bigr).
\]
Conditionally on $\F_j$, the next digits perform the local experiment of \cref{sec:onerect} in every $R$ at once: all heights and the position of every point relative to $I$ are fixed, and the next digit of each point in $I$ chooses its child uniformly and independently. Here $a^{(I)}$ is the old function and $a^{(C_1)},\dots,a^{(C_b)}$ are the child functions. So \cref{lem:local} bounds the conditional expectation of each term, and dividing by $bM$ gives \eqref{eq:drift}.
\end{proof}

\paragraph{Good transitions.}
The number $\delta_j$ is small if the points crowd into few inner regions, because a region with $m$ points contributes only $\sqrt m$. If every inner region held about $n/M$ points, the average number of points of a comparison rectangle, then $\delta_j$ would be about $\frac1{4b}\sqrt{n/M}$. We show that, when $M\le n$, typically a fraction of this that depends only on $b$ is retained at every transition.

Call an inner region \emph{heavy} if $m_R>bn/M$, that is, if it holds more than $b$ times the average number of points of a comparison rectangle. Call a transition \emph{good} if at least $n/2$ points lie in inner regions that are not heavy. Let $G_j$ be the event that the transition $j\to j+1$ is good, and let $G=\bigcap_{j<h}G_j$ be the event that all $h$ transitions are good. These events depend only on the points. We write $G_j^c$ and $G^c$ for their complements. Put
\begin{equation}\label{eq:d}
\delta_*=\frac1{8b^{3/2}}\sqrt{\frac nM}.
\end{equation}

\begin{lemma}[Good transitions]\label{lem:crowd}
Let $0\le j<h$.
\begin{enumerate}[label=(\roman*),nosep]
\item On $G_j$, we have $\delta_j\ge\delta_*$.
\item If $M\le n$, then $\Pr(G_j^c)\le(48/b)^{n/2}$, and hence $\Pr(G^c)\le h(48/b)^{n/2}$.
\end{enumerate}
\end{lemma}

\begin{proof}
(i) An inner region that is not heavy has $m_R\le bn/M$, and hence $\sqrt{m_R}\ge m_R/\sqrt{bn/M}$. On $G_j$ these regions hold at least $n/2$ points in total, so $\sum_R\sqrt{m_R}\ge\frac n2\sqrt{M/(bn)}$. By the definition of $\delta_j$ in \eqref{eq:drift},
\[
\delta_j\ \ge\ \frac1{4bM}\cdot\frac n2\sqrt{\frac M{bn}}=\frac1{8b^{3/2}}\sqrt{\frac nM}=\delta_*.
\]

(ii) Suppose that $G_j$ fails, that is, fewer than $n/2$ points lie in inner regions that are not heavy. Every other point lies in a heavy inner region or in a \emph{boundary strip}, by which we mean the bottom or the top old cell of a comparison rectangle. So more than $n/2$ points lie in the union $W$ of the boundary strips and the heavy regions. There are fewer than $M/b$ heavy regions, since no point lies in two inner regions and each heavy one contains more than $bn/M$ of the $n$ points.

Which regions are heavy depends on the positions of the points, so $W$ is not a fixed set. We therefore consider every collection $\mathcal S$ of fewer than $M/b$ inner regions, whether or not its regions are heavy. Let $W_{\mathcal S}$ be the union of the boundary strips and of the regions in $\mathcal S$, and let $E_{\mathcal S}$ be the event that at least $n/2$ points lie in $W_{\mathcal S}$. We have just seen that if $G_j$ fails, then $E_{\mathcal S}$ occurs for the collection $\mathcal S$ of heavy regions. So $G_j^c$ is contained in the union of the events $E_{\mathcal S}$.

Fix $\mathcal S$. The boundary strips are two of the $b$ old cells of every comparison rectangle, so they cover an area of $2/b$. The regions in $\mathcal S$ have area less than $1/M$ each, so they cover an area of less than $1/b$. Hence $W_{\mathcal S}$ is a fixed set of area at most $3/b$. If $E_{\mathcal S}$ occurs, then there is a set of $\lceil n/2\rceil$ indices whose points all lie in $W_{\mathcal S}$. There are at most $2^n$ sets of indices. The $n$ points are independent and uniform in the unit square, so the points of a given set all lie in the fixed set $W_{\mathcal S}$ with probability at most $(3/b)^{\lceil n/2\rceil}\le(3/b)^{n/2}$, since $3/b\le1$. So $\Pr(E_{\mathcal S})\le2^n(3/b)^{n/2}$. There are at most $2^M$ collections of inner regions, so $\Pr(G_j^c)\le2^{M+n}(3/b)^{n/2}$. Since $M\le n$, we have $2^{M+n}\le4^n=16^{n/2}$, which gives $(48/b)^{n/2}$. A union bound over the $h$ transitions gives the bound on $\Pr(G^c)$.
\end{proof}

\paragraph{Fluctuations.}
It remains to bound the fluctuation term for a fixed coloring.

\begin{lemma}[Fluctuations]\label{lem:fluct}
Fix a coloring, and let $T=\sum_{j<h}\bigl(Z_{j+1}-\E[Z_{j+1}\mid\F_j]\bigr)$ be the fluctuation term of \eqref{eq:roadmap}. Then, for every $\lambda>0$,
\begin{equation}\label{eq:fluct}
\Pr(T\le-\lambda)\le\exp\Bigl(-\frac{\lambda^2M^2}{2hn}\Bigr).
\end{equation}
\end{lemma}

\begin{proof}
Write $T_j=Z_{j+1}-\E[Z_{j+1}\mid\F_j]$, so that $T=\sum_{j<h}T_j$. The proof has three steps. In the first two, the transition $j\to j+1$ is fixed.

\emph{Step 1: one digit.} Changing the next digit of one point $P_i$ changes $Z_{j+1}$ by at most $1/M$. Indeed, the change moves $P_i$ to another child of its stage-$j$ interval $I$. So the weight of $P_i$ changes only in the $b$ children of $I$, in each of them by a number in $[-1,1]$. Hence the endpoint average of each child changes by a multiple of the function $y\mapsto\1\{V_i\le y\}$, with a factor in $[-1,1]$, and the endpoint average of no other interval of stage $j+1$ changes. Consider a cell $C\times B$ of stage $j+1$, where $C$ is a child of $I$. If $V_i\notin B$, the added function is constant on $B$, and $\osc_Ba^{(C)}$ does not change, by shift invariance. If $V_i\in B$, the added function has oscillation at most $1$ on $B$, and $\osc_Ba^{(C)}$ changes by at most $1$, by the reverse triangle inequality. The second case occurs for $b$ cells, one for each child. Since $Z_{j+1}$ is the sum of the oscillations over the $bM$ cells, divided by $bM$, it changes by at most $b/(bM)=1/M$.

\emph{Step 2: one transition.} The potential $Z_{j+1}$ is a bounded function of the $n$ next digits and of $\F_j$-measurable data, namely the heights and the first $j$ digits. The next digits are independent of each other and of $\F_j$. So Step~1 and the bounded-differences inequality (\cref{fact:bd} in the appendix, with $\Gamma$ the $\F_j$-measurable data, $\xi_1,\dots,\xi_n$ the next digits, $c=1/M$, and $-t$ in place of $t$) give, for every $t>0$,
\[
\E\bigl[e^{-tT_j}\bigm|\F_j\bigr]\le\exp\Bigl(\frac{nt^2}{2M^2}\Bigr).
\]

\emph{Step 3: all transitions.} For every $j$, the earlier terms $T_0,\dots,T_{j-1}$ are $\F_j$-measurable. So we can condition successively on $\F_{h-1},\F_{h-2},\dots,\F_0$ and apply Step~2 each time. Multiplying the $h$ bounds gives $\E e^{-tT}\le\exp(hnt^2/(2M^2))$. By Markov's inequality,
\[
\Pr(T\le-\lambda)=\Pr\bigl(e^{-tT}\ge e^{t\lambda}\bigr)\le\exp\Bigl(-t\lambda+\frac{hnt^2}{2M^2}\Bigr),
\]
and the choice $t=\lambda M^2/(hn)$ gives the claim.
\end{proof}

\begin{proposition}[Simultaneous probability bound]\label{prop:sim}
If $M\le n$, then
\begin{equation}\label{eq:sim}
\Pr\biggl[\exists\chi\in\pmone^n:\ \Fnorm{F_\chi}\le\frac h{64b^{3/2}}\sqrt{\frac nM}\biggr]\ \le\ h(48/b)^{n/2}+2^ne^{-hM/(2048b^3)}.
\end{equation}
\end{proposition}

\begin{proof}
Fix a coloring. Suppose that $G$ holds and $\Fnorm F\le h\delta_*/8$. On $G$, every conditional gain is almost surely at least $\delta_*$, by \cref{lem:drift,lem:crowd}. So \eqref{eq:roadmap} and $Z_0\ge0$ give $Z_h\ge h\delta_*+T$. Also $Z_h\le h\delta_*/4$ by \eqref{eq:Zbound}. Hence $T\le-3h\delta_*/4$, and in particular $T\le-\lambda$ for $\lambda=h\delta_*/4$. We do not optimize the constants, and apply \cref{lem:fluct} with this $\lambda$. Since $\delta_*^2=n/(64b^3M)$ by \eqref{eq:d}, the exponent in \eqref{eq:fluct} is then
\[
\frac{\lambda^2M^2}{2hn}=\frac{h\,\delta_*^2M^2}{32n}=\frac{hM}{2048b^3},
\]
and therefore $\Pr\bigl[G\text{ and }\Fnorm F\le h\delta_*/8\bigr]\le\exp(-hM/(2048b^3))$. The event $G$ does not depend on the coloring. So we bound the probability of $G^c$ once, by \cref{lem:crowd}, and take a union bound over the colorings only for the other event:
\[
\Pr\bigl[\exists\chi:\ \Fnorm{F_\chi}\le h\delta_*/8\bigr]\ \le\ \Pr(G^c)+\sum_{\chi\in\pmone^n}\Pr\bigl[G\text{ and }\Fnorm{F_\chi}\le h\delta_*/8\bigr].
\]
Since $h\delta_*/8=h\sqrt{n/M}/(64b^{3/2})$, this gives \eqref{eq:sim}.
\end{proof}

\subsection{Parameters}

It remains to choose $b$ and $h$. A comparison rectangle holds $n/M$ points on average. We take $h$ as large as possible with $n/M\ge b^{-5}h$; then also $n/M<2b^{-4}h$. The lower bound makes the threshold in \eqref{eq:sim} at least $h^{3/2}/(64b^4)$. The upper bound makes the last exponent in \eqref{eq:sim} more than $bn/4096$, so, when $M\le n$, both terms on the right-hand side of \eqref{eq:sim} are small for large $b$. This is the reason for the power $b^{-5}$: with $b^{-4}$, the last exponent would only be more than $n/4096$, which is not enough against the $2^n$ colorings. The constant $c_A$ of the theorem is $1/(64b^4)$, rounded down to $b^{-8}$, times the factor $(1+\log_2b)^{-3/2}$ that converts $h^{3/2}$ into $(\log_2n)^{3/2}$.

\begin{proof}[Proof of \cref{thm:main}]
Choose an integer $b$ so large that
\begin{equation}\label{eq:b}
\frac12\log(b/48)\ge A+2,\qquad\frac b{4096}\ge A+2+\log2.
\end{equation}
Such a choice depends only on $A$, and it forces $b>48$, so that $64\le b^4$. Set $c_A=b^{-8}/(1+\log_2b)^{3/2}$. Let $h\ge1$ be the largest integer with $b^{-5}hb^{h-1}\le n$, and set $M=b^{h-1}$. Such an $h$ exists, since $h=1$ satisfies the condition and the left-hand side tends to infinity with $h$. By maximality, $n<b^{-5}(h+1)b^h$. With $h+1\le2h$ and $b^h=bM$ this gives
\begin{equation}\label{eq:hM}
b^{-5}hM\le n<2b^{-4}hM,
\end{equation}
and with $h+1\le2^h$ and $b^{-5}\le1$ it gives $n<(2b)^h$, that is,
\begin{equation}\label{eq:logn}
\log_2n<(1+\log_2b)h.
\end{equation}
By \eqref{eq:logn} and the definition of $c_A$,
\begin{equation}\label{eq:cA}
c_A(\log_2n)^{3/2}<b^{-8}h^{3/2}.
\end{equation}

\Cref{prop:sim} needs $M\le n$. If $M>n$, then \eqref{eq:hM} gives $b^{-5}h<1$, that is, $h<b^5$, and \eqref{eq:cA} gives $c_A(\log_2n)^{3/2}<b^{-8}\,b^{15/2}<1$. So it suffices to find an anchored rectangle that has imbalance $1$ under every coloring. The rectangle $[0,1]\times[0,\min_iV_i]$ is one, because it contains just the lowest point, which is unique since all heights are distinct. So in this case the claim of the theorem holds with probability one.

So assume that $M\le n$. Then \cref{prop:sim} applies. Also $h\le2^{h-1}\le M\le n$. By the first condition in \eqref{eq:b}, the first term on the right-hand side of \eqref{eq:sim} is at most $he^{-(A+2)n}\le ne^{-(A+2)n}$. By \eqref{eq:hM} and the second condition, $hM/(2048b^3)>bn/4096\ge(A+2+\log2)n$, so the second term is at most $2^ne^{-(A+2+\log2)n}$, which equals $e^{-(A+2)n}$. The sum of the two terms is at most $(n+1)e^{-(A+2)n}\le e^{-An}$, since $n+1\le e^{2n}$. The threshold in~\eqref{eq:sim} satisfies
\[
\frac h{64b^{3/2}}\sqrt{\frac nM}\ \ge\ \frac h{64b^{3/2}}\sqrt{\frac h{b^5}}=\frac{h^{3/2}}{64b^4}\ \ge\ b^{-8}h^{3/2}\ >\ c_A(\log_2n)^{3/2}.
\]
The first inequality uses \eqref{eq:hM}, the second uses $64\le b^4$, and the last one is \eqref{eq:cA}. So by \eqref{eq:sim}, with probability at least $1-e^{-An}$, every coloring $\chi$ satisfies $\Fnorm{F_\chi}>c_A(\log_2n)^{3/2}$.
\end{proof}

\section{Related Work and Open Problems}\label{sec:related}

\paragraph{Discrepancy lower bounds.}
Matou\v{s}ek and Nikolov~\cite{MN15} used factorization norms~\cite{MNT20} to prove \mbox{$\Delta_d(n)=\Omega_d(\log^{d-1}n)$} for every fixed dimension. Wei and Yi~\cite{WY13} proved a planar lower bound of $\Omega(\log n)$ for binary nets by the orthogonal-function method. These point sets have low continuous discrepancy. The survey of Bansal and Nikolov~\cite{BN26} discusses the gap between the previously known planar bounds. Orthogonal functions and Riesz products have long been used for discrepancy lower bounds, including Hal\'asz's refinement of Roth's method~\cite{Hal81}; see the survey of Bilyk and Lacey~\cite{BL13} for their connection with the Brownian-sheet small-ball problem.

\paragraph{The Brownian sheet and Talagrand's proof.}
For the standard Brownian sheet $W$ on $[0,1]^2$, the sharp estimate is $-\log\Pr(\Fnorm W\le\eps)=\Theta\bigl(\eps^{-2}\log^3(1/\eps)\bigr)$ as $\eps\downarrow0$. The upper bound on the left-hand side is due to Bass~\cite{Bas88}, and the matching lower bound is due to Talagrand~\cite{Tal94}. Substituting $\eps=\tau/\sqrt n$ and asking for an exponent comparable to the coloring entropy $n\log2$ gives $\frac n{\tau^2}\log^3(\sqrt n/\tau)=\Theta(n)$, that is, $\tau=\Theta(\log^{3/2}n)$. This calculation only predicts the scale. The process $F_\chi$ is not Gaussian, and a Gaussian limit does not control probabilities this small.

Talagrand's proof develops a deterministic recursion based on averaging in one coordinate and selecting representatives in the other, first for Haar functions and then for the piecewise-linear functions used in the Brownian expansion \cite[Proposition~3.2 and Theorem~3.1]{Tal94}. Our local transition uses the same opposing coordinate scales. Averaging the child functions preserves the old function in expectation, while independent symmetric jumps from the points of the inner region supply the gain. The endpoint average is the conditional version of the horizontal average in Talagrand's representative-point induction \cite[proof of Proposition~3.2]{Tal94}: for an interval $I$ of stage $j$ and length $|I|$, we have $a^{(I)}(y)=\E\bigl[\frac1{|I|}\int_IF(x,y)\,dx\bigm|\F_j\bigr]$, because a point in $I$ has integral weight $(I^+-U_i)/|I|$, with conditional mean $1/2$, and points to the left or right have weight one or zero. The global argument here refines the same digit filtration point by point to obtain bounded martingale differences. A union bound over the possible sets of heavy regions shows that the total conditional gain is large, before the martingale tail estimate is applied.

\paragraph{Random inputs and applications.}
Random inputs have also been studied in online discrepancy. Jiang, Kulkarni, and Singla~\cite{JKS19} considered uniform random points for intervals and stripes, and Bansal, Jiang, Meka, Singla, and Sinha~\cite{BJMSS21} gave online algorithms for Tusn\'ady's problem with stochastic arrivals. For uniform random arrivals in $[0,1]^d$, every online algorithm incurs discrepancy $\Omega_d(\log^{d/2}T)$ within the first $T$ arrivals, with constant probability~\cite{BJSS20}. \Cref{thm:main} holds for every coloring, online or offline, and raises this bound to $\Omega(\log^{3/2}T)$ in the plane. Discrepancy is also related to range searching and approximate range counting~\cite{Lar14,WY13}. The random hard instances from \cref{thm:main} may be useful for average-case lower bounds in these settings.

\paragraph{Open problem 1: explicit constructions.}
Our proof gives a probabilistic construction of hard point sets. Is there an explicit family with discrepancy $\Omega(\log^{3/2}n)$? The main issue is to retain the oscillation gained across many scales while replacing the independent digits by a deterministic choice.

\paragraph{Open problem 2: higher dimensions.}
In dimension $d\ge3$, the gap between the bounds $\Omega_d(\log^{d-1}n)$ \cite{MN15} and $O_d(\log^{d-1/2}n)$ \cite{Nik17} for boxes in $\R^d$ remains. Bansal~\cite[Conjecture~1.4.1]{Ban22} conjectured that the discrepancy of Tusn\'ady's problem in $d$ dimensions is $O(\log^{d-1}n)$; \cref{thm:main} shows that this fails for $d=2$. Our one-coordinate refinement confines a point's influence to a single vertical interval. An extension needs an analogous local comparison when several coordinates are refined or coarsened. The Brownian-sheet conjecture predicts the exponent $\eps^{-2}\log^{2d-1}(1/\eps)$ for the negative log small-ball probability \cite[\mbox{Conjecture~4}]{BL13}. Setting $\eps=\tau/\sqrt n$ and balancing this exponent against $n\log2$ suggests $\tau=\Theta(\log^{d-1/2}n)$, matching the known upper bound. This heuristic does not establish a transfer to finite point sets.

\appendix
\section{Symmetric Sums and Bounded Differences}\label{app:facts}

\subsection{Sums of symmetric variables}\label{app:sym}

\begin{proof}[Proof of \cref{fact:sym}]
Write $X=X_1+\dots+X_m$.

\emph{The shift does not hurt.} The $X_i$ are independent and symmetric. So $X$ and $-X$ have the same distribution, and $\E|\theta+X|=\E|\theta-X|$. Averaging the two and using $|\theta+x|+|\theta-x|\ge2|x|$ gives $\E|\theta+X|\ge\E|X|$.

\emph{Two moments.} Every $X_i$ is symmetric and has a finite second moment, so $\E X_i=0$. With independence, this gives $\E X^2=\sum_i\E X_i^2=m\sigma^2$. In the expansion of $X^4$, a term has expectation zero unless every index appears an even number of times. The surviving terms have either four equal indices or two pairs of equal indices, and each pair $i<j$ arises from $\binom42=6$ terms. So
\[
\E X^4=\sum_i\E X_i^4+6\sum_{i<j}\E X_i^2\,\E X_j^2\le3m\sigma^4+3m(m-1)\sigma^4=3m^2\sigma^4.
\]

\emph{From moments to the absolute value.} The claim is trivial for $m=0$, so let $m\ge1$. By H\"older's inequality with exponents $3/2$ and $3$,
\[
\E X^2=\E\bigl(|X|^{2/3}\,|X|^{4/3}\bigr)\le(\E|X|)^{2/3}(\E X^4)^{1/3}.
\]
Hence
\[
\E|X|\ge\frac{(\E X^2)^{3/2}}{(\E X^4)^{1/2}}\ge\frac{m^{3/2}\sigma^3}{\sqrt3\,m\sigma^2}=\sigma\sqrt{m/3}.\qedhere
\]
\end{proof}

\subsection{Bounded differences}\label{app:bd}

The following is a moment-generating-function form, with a weaker constant, of the bounded-differences inequality~\cite{McD89}. Its proof is the martingale argument of Azuma~\cite{Azu67}.

\begin{fact}[Bounded differences]\label{fact:bd}
Let $\xi_1,\dots,\xi_n$ be independent random variables, and let $\Gamma$ be a random element independent of $(\xi_1,\dots,\xi_n)$. Let $f=\varphi(\Gamma,\xi_1,\dots,\xi_n)$ be integrable, where $\varphi$ is measurable and, for every value $\gamma$ of $\Gamma$, changing the value of one $\xi_i$ changes $\varphi(\gamma,\cdot)$ by at most $c>0$. Then, for every real $t$, almost surely
\[
\E\Bigl[\exp\bigl(t\,(f-\E[f\mid\Gamma])\bigr)\Bigm|\Gamma\Bigr]\le\exp\Bigl(\frac{nt^2c^2}2\Bigr).
\]
\end{fact}

\begin{proof}
Since $\Gamma$ is independent of the $\xi_i$, conditioning on $\Gamma=\gamma$ amounts to fixing the first argument of $\varphi$, by Fubini's theorem. So it suffices to prove, for each fixed $\gamma$, the unconditional bound \mbox{$\E\exp(t(f-\E f))\le\exp(nt^2c^2/2)$} for the function $f=\varphi(\gamma,\cdot)$ of $\xi_1,\dots,\xi_n$. So fix $\gamma$. Any two values of $f$ differ by at most $nc$, so $f$ is bounded.

\emph{Martingale differences.} Reveal $\xi_1,\dots,\xi_n$ one at a time. Let $f_i(s_1,\dots,s_i)$ be the conditional expectation of $f$ given $\xi_1=s_1,\dots,\xi_i=s_i$, and put
\[
Y_i=f_i(\xi_1,\dots,\xi_i)-f_{i-1}(\xi_1,\dots,\xi_{i-1}).
\]
The $Y_i$ sum to $f-\E f$. Since $\xi_i$ is independent of the earlier variables, $f_{i-1}(s_1,\dots,s_{i-1})$ is the mean of $f_i(s_1,\dots,s_{i-1},\xi_i)$. So $Y_i$ has conditional mean zero given $\xi_1,\dots,\xi_{i-1}$. Since the later variables are independent of the first $i$, the function $f_i$ is the average of $f$ over the later variables, and changing $s_i$ changes $f_i$ by at most $c$. Thus $f_i(\xi_1,\dots,\xi_i)$ differs by at most $c$ from every value of $f_i(\xi_1,\dots,\xi_{i-1},\cdot)$, and hence from their mean $f_{i-1}$. So $|Y_i|\le c$.

\emph{One difference.} Let $\mathcal H$ be a $\sigma$-field, and let $\E[Y\mid\mathcal H]=0$ and $|Y|\le c$. For every real $t$, the convexity of $y\mapsto e^{ty}$ on $[-c,c]$ gives
\[
e^{tY}\le\frac{c+Y}{2c}\,e^{tc}+\frac{c-Y}{2c}\,e^{-tc},\qquad\text{hence}\qquad\E[e^{tY}\mid\mathcal H]\le\cosh(tc)\le e^{t^2c^2/2}.
\]
The last inequality compares the Taylor series $\cosh x=\sum_kx^{2k}/(2k)!$ and $e^{x^2/2}=\sum_kx^{2k}/(2^kk!)$ term by term, using $(2k)!\ge2^kk!$.

\emph{All differences.} Apply the last bound to $Y_n$ given $\xi_1,\dots,\xi_{n-1}$, then to $Y_{n-1}$ given $\xi_1,\dots,\xi_{n-2}$, and so on. This is possible because $Y_1,\dots,Y_{i-1}$ are functions of $\xi_1,\dots,\xi_{i-1}$. Multiplying the $n$ bounds gives the claim.
\end{proof}

\section*{AI methodology statement}
The author has thought about this problem since 2013. Starting in July 2026, the author worked on it interactively with GPT (GPT-5.6, GPT-6, Codex) and Claude (Claude Opus~4.6, Claude Opus~5.5, Claude Cowork). This interaction first produced a proof based on a killed suffix walk, and then simplified it to the proof by local oscillation presented in \AIsections. The same tools, together with Xean~\cite{Xu26}, were used to check the proof and to prepare the manuscript. A Lean formalization of the main lower-bound theorem was produced with GPT and is available at \url{https://github.com/wzskytop/tusnady-lean}. The author verified the correctness and originality of all content.

\end{document}